\newif
\ifshortver
\shortvertrue
\shortverfalse

\ifshortver
\documentclass[conference]{IEEEtran}
\else
\documentclass[english, onecolumn]{IEEEtran}
\fi
\IEEEoverridecommandlockouts

\usepackage{cite}
\usepackage{algorithmic}
\usepackage{graphicx}
\usepackage{textcomp}
\usepackage{xcolor}
\def\BibTeX{{\rm B\kern-.05em{\sc i\kern-.025em b}\kern-.08em
    T\kern-.1667em\lower.7ex\hbox{E}\kern-.125emX}}

\usepackage{amssymb,amsthm}
\usepackage{amsfonts} 
\usepackage{tikz}
\usepackage{bbm}
\usetikzlibrary{positioning,arrows,shapes,chains,fit,scopes}
\usepackage{pgfplots}
\usepackage{pgfplotstable}	
\usepackage{booktabs}
\usepackage{colortbl}
\usepackage{color}
\usepackage{graphicx}
\usepackage{placeins}
\usepackage{float}
\usepackage{subfigure}
\usepackage{pgfplots}
\usepackage{pgfplotstable}	
\usepackage{amssymb}
\usepackage{color}

\usepackage{dsfont}
\usepackage{blkarray, bigstrut}
\usepackage{amsmath}
\allowdisplaybreaks 

\ifshortver
\else
\usepackage{hyperref}
\fi

\DeclareMathOperator*{\esssup}{ess\,sup}

\theoremstyle{theorem}
\newtheorem{theorem}{Theorem}
\newtheorem{lemma}[theorem]{Lemma}

\newtheorem{proposition}[theorem]{Proposition}
\theoremstyle{definition}
\newtheorem{definition}{Definition}

\theoremstyle{remark}
\newtheorem{remark}{Remark}
\newtheorem{claim}{Claim}

 \usepackage{url}

\begin{document}

\title{On the Generalization Error of Differentially Private Algorithms via Typicality
}
\ifshortver
\else
\author{%
  \IEEEauthorblockN{Yanxiao Liu\IEEEauthorrefmark{2},
    Chun Hei Michael Shiu\IEEEauthorrefmark{4},
    Lele Wang\IEEEauthorrefmark{4}
    and Deniz G\"und\"uz\IEEEauthorrefmark{2}} \\
  \IEEEauthorblockA{\IEEEauthorrefmark{2}%
             Department of Electrical and Electronic Engineering, Imperial College London, UK,\\
             yliu25@ic.ac.uk, d.gunduz@imperial.ac.uk
             }\\
  \IEEEauthorblockA{\IEEEauthorrefmark{4}%
             Department of Electrical and Computer Engineering, University of British Columbia, Canada \\
             shiuchm@ece.ubc.ca, lelewang@ece.ubc.ca}
}
\fi
\maketitle

\begin{abstract}

We study the generalization error of stochastic learning algorithms from an information-theoretic perspective, with particular emphasis on deriving sharper bounds for differentially private algorithms.
It is well known that the generalization error of stochastic learning algorithms can be bounded in terms of mutual information and maximal leakage, yielding average and high-probability guarantees, respectively.
In this work, we further upper bound mutual information and maximal leakage by explicit, easily computable formulas, using typicality-based arguments and exploiting the stability properties of private algorithms.
In the first part of the paper, we strictly improve the mutual information bounds by Rodríguez-Gálvez et al.\ (IEEE Trans.\ Inf.\ Theory, 2021).
In the second part, we derive new upper bounds on maximal leakage of learning algorithms.
In both cases, the resulting bounds on the information measures translate directly into generalization error guarantees.
\begin{IEEEkeywords}
Generalization error, differential privacy, maximal leakage, mutual information, typicality. 
\end{IEEEkeywords}
\end{abstract}

\section{Introduction}
\label{sec::intro}

A randomized statistical learning algorithm $P_{W|S}$ can be defined as a probabilistic mapping from a training dataset $S := \allowbreak \left(Z_1,\ldots,Z_N \right)\in\mathcal{Z}^N$ consisting of $N$ independent and identically distributed (i.i.d.) data points from distribution $P_Z$ (which is unknown) on the space $\mathcal{Z}$ to a hypothesis $W$ from the hypothesis space $\mathcal{W}$. 
Similar to~\cite{rodriguez2021upper}, we assume $\mathcal{Z}$ is finite and $\mathcal{W}$ is discrete. 
The generalization performance of this algorithm is measured by the difference between the expected loss $\mathbf{E}_{P_Z}\left[\ell(W,Z) \right]$ and the empirical loss $\frac{1}{N}\sum^N_{i=1} \ell(W,Z_i)$, where $\ell(\cdot, \cdot)$ is the loss function. 
We assume $\ell(\cdot, \cdot)$ is $\sigma$-sub-Gaussian\footnote{A random variable (r.v.) bounded on $[a,b]$ is $\frac{b-a}{2}$-sub-Gaussian. 
If $(X_i)_i$ are independent $\sigma$-sub-Gaussian r.v.s, $\frac{1}{n}\sum^n_{i=1}X_i$ is $\frac{\sigma}{\sqrt{n}}$-sub-Gaussian. } in this paper, i.e., for all $\lambda\in\mathbb{R}$ and all $w\in\mathcal{W}$,
\begin{equation*}
    \mathbf{E}_{Z\sim P_Z} \big[
    \exp\big(  \lambda
    (\mathbf{E}[\ell(w,Z)] -\ell(w, Z) )
    \big)
    \big]\leq \exp\big( \lambda^2\sigma^2/2  \big). 
\end{equation*}
    
One of the central challenges in statistical learning theory is to guarantee that a learning algorithm $\mathcal{A}$ can ``generalize well'', i.e., the \emph{generalization error} defined as 
\begin{equation}
    \mathrm{gen}(W,S) := \mathbf{E}_{P_Z}\left[\ell(W,Z) \right] - \frac{1}{N}\sum^N_{i=1} \ell(W,Z_i)
    \label{eq::def_gen_error}
\end{equation} 
is bounded; and hence how much $\mathcal{A}$ overfits is under control.

A vast array of studies have been proposed for this objective. 
To guarantee that the trained algorithm can generalize well, the theory of uniform convergence~\cite{vapnik2015uniform} states that the output has to be ``sufficiently simple''. 
Similar to this idea, bounds in terms of the VC-dimension and the Rademacher complexity have been discussed~\cite{vapnik1998statistical, boucheron2005theory}. 
However, they make no reference to the algorithm. 
To make a specific algorithm generalize, various types of methods have been studied, e.g., based on compression schemes~\cite{littlestone1986relating} or uniform stability~\cite{bousquet2002stability}. 
Motivated by controlling the bias in data analysis using the \emph{mutual information} between the algorithm input and output, it was shown by~\cite{russo2016controlling, xu2017information} that the expectation of the generalization error can be bounded by the mutual information $I(S; W)$ between the training dataset $S$ and the output hypothesis $W$: 
\begin{equation}
    \label{eq::gen_bd_expt_MI}
    \mathbf{E}\left[
    \left|\mathrm{gen}(W,S)\right|
    \right] 
    \leq \sqrt{\frac{2\sigma^2}{N} I(S; W)}.
\end{equation}

This provides the intuition that ``an algorithm that leaks little information about the dataset generalizes well''~\cite{banerjee2021information, bassily2018learners}, which partially coincides with the property of \emph{private} algorithms. 
There exist various notions of \emph{privacy}, e.g., total variation distance~\cite{rassouli2019optimal}, max-information~\cite{dwork2015generalization} and differential privacy~\cite{dwork2006calibrating, kotz2012laplace}. 
See Section~\ref{subsec::DP} for a detailed discussion.

Maximal leakage $\mathcal{L}(S \rightarrow W)$~\cite{issa2019operational} of an algorithm $P_{W|S}$, as an information leakage measure, has been discussed recently~\cite{liao2017hypothesis, saeidian2023pointwise, esposito2021generalization, liao2018tunable}. 
It has been shown by~\cite{esposito2021generalization} that the generalization error can be bounded by $\mathcal{L}(S \rightarrow W)$: given $\eta \in (0,1)$, 
\begin{equation}
    \mathbf{P}\left(\big|\mathrm{gen}(W,S)\big|\geq\eta \right) 
    \leq 2 \exp\left(
    \mathcal{L}(S\rightarrow W) - \frac{N \eta^2}{2\sigma^2}
     \right). 
     \label{eq::gen_bd_prob_ML}
\end{equation}

One major meaningful difference between~\eqref{eq::gen_bd_expt_MI} and~\eqref{eq::gen_bd_prob_ML} is that~\eqref{eq::gen_bd_prob_ML} is a \emph{high-probability} (a.k.a.\ \emph{single-draw}~\cite{catoni2007pac, hellstrom2020generalization}) bound; that is, it admits an \emph{exponentially decaying} (in the size of the dataset $N$) bound on the probability of having a large generalization error, which is generally untrue for mutual-information-based bounds that only guarantee a linearly decaying bound~\cite{bassily2018learners}. 
Moreover, \eqref{eq::gen_bd_prob_ML} depends on the samples through its support only and hence independent from the distribution, this can be useful in analysis of additive-noise settings~\cite{issa2023generalization}.

Generalization error bounds in terms of other information measures can be found in~\cite{steinke2020reasoning, sefidgaran2022rate, sefidgaran2024data, esposito2021generalization, hellstrom2020generalization, chu2023unified, wang2024generalization, liu2026tighter}. 
The connection between privacy measures and generalization has been discussed in~\cite{dwork2015generalization, rogers2016max, bassily2016algorithmic, esposito2021generalization, issa2023generalization}. 
Privacy is also tightly related to compression~\cite{chen2020breaking, liu2024universal}, and compressibility can be used to bound the generalization error as well~\cite{littlestone1986relating, sefidgaran2022rate, sefidgaran2024data}. 
In~\cite{rodriguez2021upper}, an approach based on \emph{typicality}~\cite{cover1999elements} has been used to bound the generalization error of private algorithms, and their bounds are in more explicit, easily computable forms, since mutual information can sometimes be hard to calculate.

The contributions of the current paper are two-fold. 
In the first part (Section~\ref{sec::MI}), we utilize the stability property of differentially private algorithms to further strictly improve the mutual information bounds in~\cite{rodriguez2021upper}. 
In the second part (Section~\ref{sec::ML}), we derive novel upper bounds on maximal leakage, which directly translate to upper bounds on mutual information~\cite{issa2019operational}. 
We hence provide explicit generalization bounds via~\eqref{eq::gen_bd_expt_MI} and~\eqref{eq::gen_bd_prob_ML}. 

\subsection*{Notation}

We use calligraphic letters (e.g., $\mathcal{W}$), capital letters (e.g., $W$) and lower-case letters (e.g., $w$) to denote sets, random variables and instances, respectively. 
We assume all random variables are from finite alphabets and logarithm and entropy are to the base $e$ unless otherwise stated. 
For a statement $S$, $\mathds{1}_{\{S\}}$ is its indicator, i.e., $\mathds{1}_{\{S\}}$ is $1$ if $S$ holds, and $0$ otherwise. 
For distributions $P$ and $Q$, we use $P \ll Q$ to denote that $P$ is absolutely continuous with respect to $Q$.

\section{Preliminaries}
\label{sec::prelim}

We use $S := \left(Z_1,\ldots,Z_N \right)$ to denote a training dataset, which is a collection of $N$ instances $Z_i \in \mathcal{Z}$ i.i.d.\ from $P_Z$ (which is unknown). 
A stochastic learning algorithm $\mathcal{A}: \mathcal{Z}^N \rightarrow \mathcal{W}$ can be characterized by the conditional distribution $P_{W|S}$, i.e., $S$ is a random variable with distribution $P_S = P_Z^{\otimes N}$ over $\mathcal{Z}^N$, where $|\mathcal{Z}|$ is finite. 
We consider algorithms that are \emph{permutation-invariant}: $\mathcal{A}$ operates on a \textcolor{black}{(multi)set} instead of an ordered sequence; the hypothesis generated by $\mathcal{A}$ with dataset $S$ is \emph{equal in distribution} to the hypothesis generated by $\mathcal{A}$ with a \emph{permuted} dataset $\mathcal{P}(S)$, i.e., $\mathcal{A}(S) \stackrel{d}{=} \mathcal{A}(\mathcal{P}(S))$ for all permutations $\mathcal{P}$. 
This is a common assumption in the generalization analysis~\cite{bousquet2002stability, shalev2010learnability}.

In this paper, we upper-bound the generalization error~\eqref{eq::def_gen_error} via either an average bound of the form~\eqref{eq::gen_bd_expt_MI} or a high-probability bound of the form~\eqref{eq::gen_bd_prob_ML}. 
There exist other types of generalization bounds, e.g., PAC-Bayesian bounds~\cite{hellstrom2020generalization, guedj2021still, dziugaite2017computing}, which may also be of interest in the future. 
We consider $\sigma$-sub-Gaussian loss functions in this paper.

\subsection{Method of Types}

Consider $Z^N := (Z_1,\ldots,Z_N)$ where $Z_i \stackrel{\mathrm{i.i.d}}{\sim} P_Z$ for $i = 1, \ldots, N$. We can then define the \emph{type}~\cite{cover1999elements} of $Z^N$ as follows.

\begin{definition}
    The \emph{type} of a sequence $Z^N\in\mathcal{Z}^N$ is the relative proportion of occurrences of each symbol of $\mathcal{Z}$, i.e., $T_{Z^N}(t) = \mathsf{N}(t|Z^N)/N$ for all $t\in\mathcal{Z}$, where $ \mathsf{N}(t|Z^N)$ denotes the number of times symbol $t$ occurs in the sequence $Z^N$. The set of all possible types of sequences of elements from $\mathcal{Z}$ of length $N$ is denoted by $\mathcal{T}_{\mathcal{Z}, N}$. 
\end{definition}

The type and its variants were used to prove achievability results in network information theory problems, including both asymptotic~\cite{el2011network} and nonasymptotic~\cite{tan2013dispersions} settings. 
For the latter case, also see other methods~\cite{li2021unified, liu2025one, liu2025thesis, liu2025nonasymptotic}. 
Even though the number of items in the space scales exponentially in $N$, it is known that the total number of types only scales polynomially: 
\cite{cover1999elements} gives $|\mathcal{T}_{\mathcal{Z}, N}|\leq (N+1)^{|\mathcal{Z}|}$, which was improved by~\cite{rodriguez2021upper}: 
\begin{claim}[Claim 1 of~\cite{rodriguez2021upper}]
    \label{clm::rodriguez2021upper_typicality_bd}
    $|\mathcal{T}_{\mathcal{Z}, N}|\,\leq\,(N+1)^{|\mathcal{Z}| - 1}$, with equality if and only if $|\mathcal{Z}| = 2$. 
\end{claim}
The property of the type that will be utilized in this paper is that it is equal to that of any permutation of the same sequence, i.e., $T_{Z^N} = T_{\mathcal{P}(Z^N)}$ for all permutations $\mathcal P$. 
Hence, $T_{Z^N}$ can \emph{uniquely} define a dataset $S=\{Z_1,\ldots,Z_N\}$. 
We then define the distance between two datasets $S$ and $S'$ via their types.

\begin{definition}[Definition 5 of \cite{rodriguez2021upper}]
    \label{def::distance}
    The \emph{distance} between two datasets  $S,S'$ is the minimum number of instances that needs to be changed in $S$ to obtain $S'$. 
    \begin{align*}
        \mathsf{d}(S,S') 
        & \triangleq \frac{1}{2} \sum_{a\in\mathcal{Z}} \big|
        \mathsf{N}(a|S) - \mathsf{N}(a|S')
        \big| \\
        & = \frac{N}{2} \big\Vert T_S - T_{S'} \big\Vert_1
    \end{align*}
    where $\mathsf{N}(a|S)$ denotes the number of times sample $a$ occurs in the dataset $S$.
\end{definition}
Due to the permutation invariance of the assumed algorithms, the use of types is justified by the bijection between types and datasets.

\subsection{Differential Privacy}
\label{subsec::DP}

Differential privacy~\cite{dwork2006calibrating, kotz2012laplace}, as one of the most celebrated privatization schemes, has been widely studied in the past years (also see its variants~\cite{mironov2017renyi, kairouz2015composition, dong2022gaussian}). 
It is defined as follows: 
\begin{definition}
    \label{def::DP}
    Given an algorithm $\mathcal{A}$ which induces the conditional distribution $P_{W|S}$, we say it satisfies $(\varepsilon, \delta)$-DP if for any neighboring (differ in a single data point) $(s,s')$ and $\mathcal{V}\subseteq \mathcal{W}$, it holds that for all measurable $W$,
    \begin{equation*}
        \mathbf{P}\big(W\in\mathcal{V} | S=s\big)
        \leq e^\varepsilon\,
        \mathbf{P}\big(W\in\mathcal{V} | S=s'\big) + \delta.
    \end{equation*}
    In particular, when $\delta = 0$ we simply refer to it as $\varepsilon$-DP.
\end{definition}

Differentially private algorithms are \emph{stable}~\cite{bassily2016algorithmic}, i.e., they produce similar outputs on similar input datasets.
For an $\varepsilon$-DP algorithm $\mathcal{A}$, the maximum log-likelihood ratio between any two neighboring input datasets is bounded by $\varepsilon$, and the notion of \emph{neighboring datasets} relates to the distance between datasets. 
We say that two datasets are neighbors if they have distance $1$ (in terms of Definition~\ref{def::distance}). 

One useful property of DP is \emph{group privacy}~\cite[Theorem~2.2]{dwork2014algorithmic}: for an $\varepsilon$-DP algorithm $\mathcal{A}$, if two (fixed) datasets $S$ and $S'$ are at distance $k$, then the KL divergence of their output distributions is bounded by
\begin{equation}
    D_{\mathrm{KL}}\big(\mathcal{A}(S)\,\Vert\,\mathcal{A}(S') \big)\leq k\varepsilon \tanh\left(k\varepsilon/2 \right)\leq k\varepsilon.
    \label{eq::group_DP}
\end{equation}

Maximal leakage, another privacy measure on which this paper focuses, will be discussed in Section~\ref{sec::ML}.

\section{Mutual Information Bounds}
\label{sec::MI}

In this section, we present novel generalization bounds that can be efficiently computed by bounding mutual information.

This idea has been explored by~\cite{rodriguez2021upper}, which provides generalization bounds of discrete, private algorithms by bounding the relative entropy $D_{\mathrm{KL}}(P_{W|S}\,\Vert\, Q_W)$, since by the \emph{golden formula}~\cite[Section 4.1]{polyanskiy2025information} we know 
\begin{equation}
    \label{eq::MI_bd_by_KL}
    I(S; W)\leq \mathbf{E}_{s\sim P_S}\left[D_{\mathrm{KL}} (P_{W|S=s} \,\Vert\, Q_W) \right],
\end{equation}
and the equality is achieved if and only if $Q_W = P_W$. 
The expression $D_{\mathrm{KL}} (P_{W|S} \,\Vert\, Q_W)$ also  appears in other generalization error characterizations, e.g., \cite{hellstrom2020generalization}. 
These upper bounds on relative entropy can be directly translate to generalization error bounds by~\eqref{eq::gen_bd_expt_MI}. 
We first review a result in \cite{rodriguez2021upper} by Claim~\ref{clm::rodriguez2021upper_typicality_bd}:

\begin{proposition}
    \label{prop::rodriguez_prop12}
    Let $S\in\mathcal{S}$ be a dataset of $N$ instances $Z^N\in\mathcal{Z}^N$ sampled i.i.d. from $P_Z$. 
    Let also $W\in\mathcal{W}$ be a hypothesis obtained with an algorithm $\mathcal{A}$, characterized by $P_{W|S}$. 
    For all $s\in\mathcal{S}$, there exists a distribution $Q_W$ such that
    \begin{itemize}
        \item If $\mathcal{A}$ \textcolor{black}{has a discrete output distribution}, then 
        \begin{equation}
            D_{\mathrm{KL}} (P_{W|S=s} \,\Vert\, Q_W) \leq \left(|\mathcal{Z}|-1\right)\log(N+1).
            \label{eq::rodri_MI_prop1}
        \end{equation}

        \item If $\mathcal{A}$ is also $\varepsilon$-DP and $\varepsilon\leq 1$, then 
        \begin{equation}
            D_{\mathrm{KL}} (P_{W|S=s} \,\Vert\, Q_W)\leq (|\mathcal{Z}|-1)\log(1 + e\varepsilon N). \label{eq::rodri_MI_prop2}
        \end{equation}
    \end{itemize}
\end{proposition}

Note that in the regime $\epsilon > 1$, \eqref{eq::rodri_MI_prop1} is tighter than \eqref{eq::rodri_MI_prop2}.

Combining with the average generalization error bound~\eqref{eq::gen_bd_expt_MI} and the golden formula~\eqref{eq::MI_bd_by_KL}, from~\eqref{eq::rodri_MI_prop1} one can prove 
\[\mathbf{E}\left[
    \left|\mathrm{gen}(W,S)\right|
    \right] \leq \sqrt{\frac{2\sigma^2}{N} \left(
    |\mathcal{Z}|-1
    \right)\log(N+1)}
\]
and a sharper bound 
\[
\mathbf{E}\left[
    \left|\mathrm{gen}(W,S)\right|
    \right] \leq \sqrt{\frac{2\sigma^2}{N} 
    (|\mathcal{Z}|-1)\log(1 + e\varepsilon N). 
    }
\]
if $\mathcal{A}$ is also $\varepsilon$-DP with $\varepsilon \leq 1$.

\smallskip
In this paper, we first provide a slightly improved version of Claim~\ref{clm::rodriguez2021upper_typicality_bd}, since the proof of Claim~\ref{clm::rodriguez2021upper_typicality_bd} (see~\cite[Appendix~G.1]{rodriguez2021upper}) \textcolor{black}{employed the coarse relaxation} $1 + N/i\leq 1 + N$ for every $i$.

\begin{claim}
    \label{clm::typicality_bd}
    $|\mathcal{T}_{\mathcal{Z}, N}|\,\leq\,\frac{1}{\sqrt{2\pi (|\mathcal{Z}|-1)}}\left(\frac{e}{|\mathcal{Z}|-1}\left(N+\frac{|\mathcal{Z}|}{2}\right)\right)^{|\mathcal{Z}|-1}$. 
\end{claim}

\ifshortver
\else 
The proof of Claim~\ref{clm::typicality_bd} can be found in Appendix~\ref{app::typicality_bd}. 
\fi

\medskip

Claim~\ref{clm::typicality_bd} improves Claim~\ref{clm::rodriguez2021upper_typicality_bd} by a multiplicative constant factor $<1$. 
Take $k = |\mathcal{Z}|-1$, this constant factor $(k+1)/2^k$ is actually an \emph{optimal} factor $c$ such that $\binom{N+k}{k}\leq c(N+1)^k$ holds for all $N\geq 1$, 
\textcolor{black}{it leads to a strictly tighter bound as $|\mathcal{Z}|$ grows}. 
However, we note that the advantages of our theorems do \emph{not} mostly come from here, we are just using Claim~\ref{clm::typicality_bd} \textcolor{black}{to sharpen the resulting bounds.}

\begin{lemma}[Lemma 2 of~\cite{rodriguez2021upper}]
    \label{lem::rodriguez2021upper_KL_bd}
    Let $P$ and $Q$ be two probability distributions such that $P \ll Q$.
    Let $Q$ be a finite mixture of probability distributions such that $Q = \sum_b \omega_b Q_b$, where $\sum_b \omega_b = 1$, and $P \ll Q_b$ for all $b$.
    Then we have 
    \ifshortver

    \else
    \fi
    \begin{align}
    	D_{\mathrm{KL}}(P\Vert Q) &\leq - \log \left( \sum\nolimits_b \omega_b \exp \big( {-} D_{\mathrm{KL}}(P\Vert Q_b) \big) \right) 
    	\label{eq:rodriguez_lemma_ub_1} \\
    	&\leq \min_b \big\lbrace D_{\mathrm{KL}}(P\Vert Q_b) - \log(\omega_b) \big\rbrace .
    	\label{eq:rodriguez_lemma_ub_2} 
    \end{align}
\end{lemma}

Inequalities in above lemma can already be tight and cannot be uniformly improved. 
However, for the application of differentially private algorithms, we can sharpen it as follows.

\begin{lemma}
    \label{lem::dp_overlap_mixture}
    Let $Q$ be a mixture of $M$ probability distributions $(M<\infty)$ such that 
    \[
    Q = \sum_{b=1}^M \omega_b Q_b,
    \qquad \omega_b>0,\ \ \sum_{b=1}^M \omega_b = 1. 
    \]
    Fix an index $i\in\{1,\dots,M\}$. 
    For each $b$, let $\alpha_{b,i}\in [0,1]$ be the largest constant such that for every measurable set $E$,
    \begin{equation*}
        Q_b(E)\geq \alpha_{b,i}\, Q_i(E).
    \end{equation*}
    Then for any $P \ll Q_i$,  
    \begin{equation*}
        D_{\mathrm{KL}}(P\Vert Q) \leq D_{\mathrm{KL}}(P\Vert Q_i) - \log\left(
        \sum_{b=1}^M \omega_b \alpha_{b,i}
        \right).
    \end{equation*}
\end{lemma}

\ifshortver
\else 
The proof of Lemma~\ref{lem::dp_overlap_mixture} can be found in Appendix~\ref{app::proof_dp_overlap_mixture}. 
\fi

\ifshortver \else
{
\color{black}

\begin{remark}
    The constant $\alpha_{b,i}$ in Lemma \ref{lem::dp_overlap_mixture} admits the following explicit characterization,
    \begin{align*}
        \alpha_{b,i} = \begin{cases}
            \left(\esssup_{Q_b} \frac{dQ_i}{dQ_b}\right)^{-1}, & \text{if } Q_i \ll Q_b \text{ and } \esssup_{Q_b} \frac{dQ_i}{dQ_b} < \infty \\
            0, & \text{otherwise}.
        \end{cases}
    \end{align*}
    Equivalently, whenever $Q_i \ll Q_b$, $\alpha_{b,i} = \exp\left(-D_\infty(Q_i \| Q_b)\right)$, with the convention that $\alpha_{b,i} = 0$ if $D_\infty(Q_i \| Q_b) = \infty$.
\end{remark}
}
\fi

\begin{remark}
    \label{rem::dP_Zpecialization}
    If $Q_b = P_{W|S=s_b}$ are output distributions of an $\varepsilon$-DP algorithm and $\mathsf{d}(\cdot,\cdot)$ is the dataset distance,
    then group privacy implies that for all measurable $E$,
    \begin{align*}
         Q_b(E) & = P_{W|S=s_b}(E) \\
         & \geq e^{-\varepsilon  \mathsf{d}(s_b,s_i)}\,P_{W|S=s_i}(E) \\
         & = e^{-\varepsilon  \mathsf{d}(s_b,s_i)}\,Q_i(E).
    \end{align*}
    Hence Lemma~\ref{lem::dp_overlap_mixture} applies with $\alpha_{b,i}=e^{-\varepsilon  \mathsf{d}(s_b,s_i)}$ for any $b$. 
\end{remark}

\begin{remark}
    The bound by Lemma~\ref{lem::dp_overlap_mixture} is always at least as tight as the one in~\eqref{eq:rodriguez_lemma_ub_2} since $\sum^M_{b=1} \omega_b \alpha_{b,i} \geq \omega_i \alpha_{i,i} =  \omega_i$, and is strictly tighter whenever $\sum_{b\neq i} \omega_b \alpha_{b,i} > 0$. 
\end{remark}

We now consider differentially private algorithms and show how to further improve the results.

We first define the concept of \emph{representative sets}, introduced in~\cite[Appendix C]{rodriguez2021upper}, which will also be used in our proof. 
The idea is to utilize the property that a private algorithm produces statistically similar outputs for neighboring datasets.

Note that all types $T_s$ lie inside the unit hypercube $[0,1]^{|\mathcal{Z}|-1}$ as the type of a sequence can be uniquely identified by the first $(|\mathcal{Z}|-1)$ dimensions. 
Rather than discussing types, we study the vector of counts $\mathsf{N}_s$ by $\mathsf{N}_s(a) = \mathsf{N}(a|s)$ for all $a\in\mathcal{Z}$, and hence $\mathsf{N}_s$ lies in a $[0,N]^{|\mathcal{Z}|-1}$ hypercube. 
For the interval $[0,N]$, each dimension can be split into $t$ parts, where $t\in[1:N]$, resulting in a $[0,N]^{|\mathcal{Z}|-1}$ hypercube covering of $t^{|\mathcal{Z}|-1}$ smaller hypercubes, and the side of each small hypercube has length $l = N/t$ and $l' = \lfloor l\rfloor+1$ atoms.

We obtain the \emph{representative set} by taking one \emph{center atom} (one dataset) from each small hypercube in $t^{|\mathcal{Z}|-1}$ cover, by:
\begin{itemize}
    \item If $l'$ is odd, choose the central atom as the corresponding coordinate for the center of the small hypercube. 

    \item  If $l'$ is even, enlarge the small hypercube in one unit and choose the center of this enlarged hypercube.
\end{itemize}

\begin{theorem}
    \label{thm::improved_KL_bd}
    Let $\mathcal A$ be an $\varepsilon$-DP algorithm under the dataset distance $\mathsf{d} (\cdot,\cdot)$, with an induced distribution $P_{W|S}$ where $W\in \mathcal{W}$ is a hypothesis obtained by an algorithm $\mathcal{A}$. 
    For all datasets $s\in\mathcal{S}$, there exists a distribution $Q_W$ on $\mathcal{W}$ such that
    \begin{align}
    & D_{\mathrm{KL}} \bigl(P_{W|S=s}\,\|\,Q_W\bigr)
    \leq \min_{t\in \{1,2,\dots,N\}}
    \Big\{
    (|\mathcal Z|-1)\frac{\varepsilon N}{t} +
    \nonumber\\
    & \qquad\qquad\quad
     (|\mathcal Z|-1)\log t
    - \mathds{1}_{\{t\geq 2\}}\log \bigl(1+e^{-\varepsilon N} \bigr) \Big\}.
    \label{eq:prop2_overlap_not}
    \end{align}
\end{theorem}

{
\color{black}

\begin{proof}[Proof sketch]
Let $m:=|\mathcal Z|$, construct representative set $\mathcal{S}_0= \allowbreak \{s_1,\dots,s_M\}$ by partitioning the space $[0,N]^{m-1}$ into $t^{m-1}$ subcubes of side length $\frac{N}{t}$. 
Take $1$ dataset from each subcube, we know $M\le t^{m-1}$ and define $Q_W^{(t)}:=\frac1M\sum_{b=1}^M P_{W|S=s_b}$. 

Fix any dataset $s$. For $S_0$, there exists an index $i$ such that 
\[
\mathsf d(s,s_i)\le (m-1) N/t.
\]

For every $b$ and measurable $E$, group privacy of $\varepsilon$-DP gives, 
\[
P_{W|S=s_b}(E)\ge e^{-\varepsilon \mathsf d(s_b,s_i)} P_{W|S=s_i}(E).
\]
Applying Lemma~\ref{lem::dp_overlap_mixture}, we upper bound $D_{\mathrm{KL}}(P_{W|S=s}\|Q_W^{(t)})$ by 
\[
D_{\mathrm{KL}}(P_{W|S=s}\|P_{W|S=s_i})
-\log\Big(\frac1M\sum_{b=1}^M e^{-\varepsilon \mathsf d(s_b,s_i)}\Big),
\]
and the first term is bounded by $(m-1)\varepsilon N/t$ by group privacy, and the second is bounded by $\sum_{b=1}^M e^{-\varepsilon \mathsf d(s_b,s_i)} \ge 1 + e^{-\varepsilon N}$. Combining these yields, for every $s$, $D_{\mathrm{KL}}(P_{W|S=s}\|Q_W^{(t)})
\le$
\[
(m-1){\varepsilon N}/{t}
+(m-1)\log t
-\mathds 1_{\{t\ge2\}}\log(1+e^{-\varepsilon N}).
\]
\end{proof}}

\ifshortver
\else 
The proof of Theorem~\ref{thm::improved_KL_bd} can be found in Appendix~\ref{app::improved_KL_bd}. 
\fi
It is obvious that Theorem~\ref{thm::improved_KL_bd} \emph{strictly} improves~\cite[Proposition 2]{rodriguez2021upper}. 

In~\cite[Proposition~3]{rodriguez2021upper}, a different idea of counting only the necessary hypercubes was employed to further sharpen the bounds. 
Here, we adopt the similar idea and combine it with our technique to further strengthen our bounds.

\ifshortver
\else
\begin{theorem} 
\label{thm::prop3_plus_overlap}
Let $\mathcal A$ be an $\varepsilon$-DP algorithm under the distance $\mathsf d(\cdot,\cdot)$.
Fix $|\mathcal Z|\geq 2$ and $N\geq 1$.
For every dataset $s\in\mathcal{S}$, there exists a distribution $Q_W^{(t)}$ such that
\begin{align}
 & D_{\mathrm{KL}} \bigl(P_{W|S=s}\,\|\,Q_W^{(t)}\bigr) \nonumber\\
&\le \min_{t\in \{1,2,\dots,N\}}
\Big\{
(|\mathcal Z|-1)\frac{\varepsilon N}{t}
+
(|\mathcal Z|-1)\log \left(t+\frac{|\mathcal Z|-2}{2}\right) \nonumber \\
&\qquad
-\log ((|\mathcal Z|-1)!)
-\mathds 1_{\{t\geq 2\}}\log \bigl(1+e^{-\varepsilon\Delta_t}\bigr)\Big\}.
\label{eq:prop3_overlap_master_no_m} 
\end{align}
\end{theorem}
\fi

\ifshortver
\else 
The proof of Theorem~\ref{thm::prop3_plus_overlap} can be found in Appendix~\ref{app::prop3_plus_overlap}. 
\fi

Furthermore, similar to the idea of~\cite[Proposition 4]{rodriguez2021upper}, \textcolor{black}{we can employ the strong typicality to derive Theorem~\ref{thm::prop4_improved_overlap}, whose proof can be found in~\cite{liu2026typical_arxiv}. 
However, we note that Theorem~\ref{thm::prop4_improved_overlap} is not uniformly tighter than Theorem~\ref{thm::prop3_plus_overlap}, e.g., see Figure~\ref{fig::N_7_Z_6}. 
}

\begin{theorem} 
\label{thm::prop4_improved_overlap}
Let $\mathcal A$ be an $\varepsilon$-DP algorithm under the dataset distance $\mathsf d(\cdot,\cdot)$.
Fix $|\mathcal Z|\geq 2$ and $N\geq 1$.
For every dataset $s\in\mathcal{S}$, there exists a distribution $Q_W^{(t)}$ such that
\begin{align}
& I(S;W) \le \min_{t\in\{1,2,\dots,\lfloor 2\sqrt{N\log N}\rfloor\}}
\Big\{ \frac{\varepsilon |\mathcal Z|\sqrt{N\log N}}{t}
\nonumber\\
&\quad \quad \quad \quad
 -\mathds 1_{\{t\ge 2\}}\log \bigl(1+e^{-\varepsilon (\lceil 2\sqrt{N\log N}/t \rceil+1)}\bigr) + \frac{2 |\mathcal Z| \varepsilon}{N}  \nonumber \\
&\quad \quad \quad \quad
+ \mathds 1_{\{t\ge 2\}} 
\min\{ |\mathcal Z|\log t, \log(|\mathcal Z|t^{|\mathcal Z| -  1})\}\Big\}
\label{eq:prop4_improved_main}
\end{align} 
\end{theorem}

\ifshortver
\else
Note~\eqref{eq:prop4_improved_main} can be written in an equivalent form 
\begin{align*}
& I(S;W)
\le
\frac{\varepsilon |\mathcal Z|\sqrt{N\log N}}{t}
 -\mathds 1_{\{t\ge 2\}}\log \bigl(1+e^{-\varepsilon\Delta_t}\bigr) \nonumber \\
&\quad\,\,\,\,\,\,
+ \frac{2 |\mathcal Z| \varepsilon}{N}  
+ \mathds 1_{\{t\ge 2\}} 
\cdot \left(
|\mathcal{Z}|\log t - \left[\log \left(\frac{t}{|\mathcal{Z}|} \right)\right]_+
\right)
\end{align*} 
by considering $\min\{x,y\} = x - [x-y]_+$ where $[u]_+ := \max\{u, 0\}$. 
We choose to present~\eqref{eq:prop4_improved_main} in Theorem~\ref{thm::prop4_improved_overlap} since the meaning of each term is clearer. 
\fi
\medskip

We plot two representative scenarios and show comparison between our bounds and the bounds in~\cite{rodriguez2021upper}. 
\begin{enumerate}
    \item Figure~\ref{fig::N_3_Z_binary}: A \emph{small} dataset with $N=1\times 10^{3}$ and $|\mathcal{Z}|=2$. 
    The setting is typical in applications with limited sample sizes and binary outcomes/features, such as clinical studies or small A/B tests, 
    where we collect data from thousands of patients/users and each record carries a binary attribute (e.g., responder/non-responder, click/no-click). 
    Figure~\ref{fig::N_3_Z_binary} shows that Theorem~\ref{thm::improved_KL_bd}, \ref{thm::prop3_plus_overlap} and~\ref{thm::prop4_improved_overlap} strictly improve Proposition $2$, $3$ and $4$ in~\cite{rodriguez2021upper}, respectively, and our Theorem~\ref{thm::prop4_improved_overlap} provides the tightest guarantee.

    \item Figure~\ref{fig::N_7_Z_6}: A \emph{large} dataset with $N=1\times 10^{7}$ and $|\mathcal{Z}|=1\times 10^{6}$. 
    This regime arises in large-scale telemetry or interaction logs with high-cardinality categorical values, such as search queries, URLs, item IDs, or token IDs, 
    where one observes tens of millions of events drawn from a vocabulary on the order of millions. 
    Figure~\ref{fig::N_7_Z_6} shows that in this case we only substantially  improve the bounds in~\cite{rodriguez2021upper} when the privacy budget $\epsilon$ is small. 
\end{enumerate}

\ifshortver
\else

\begin{figure}[htpb]
    \centering
    \includegraphics[scale=0.3]{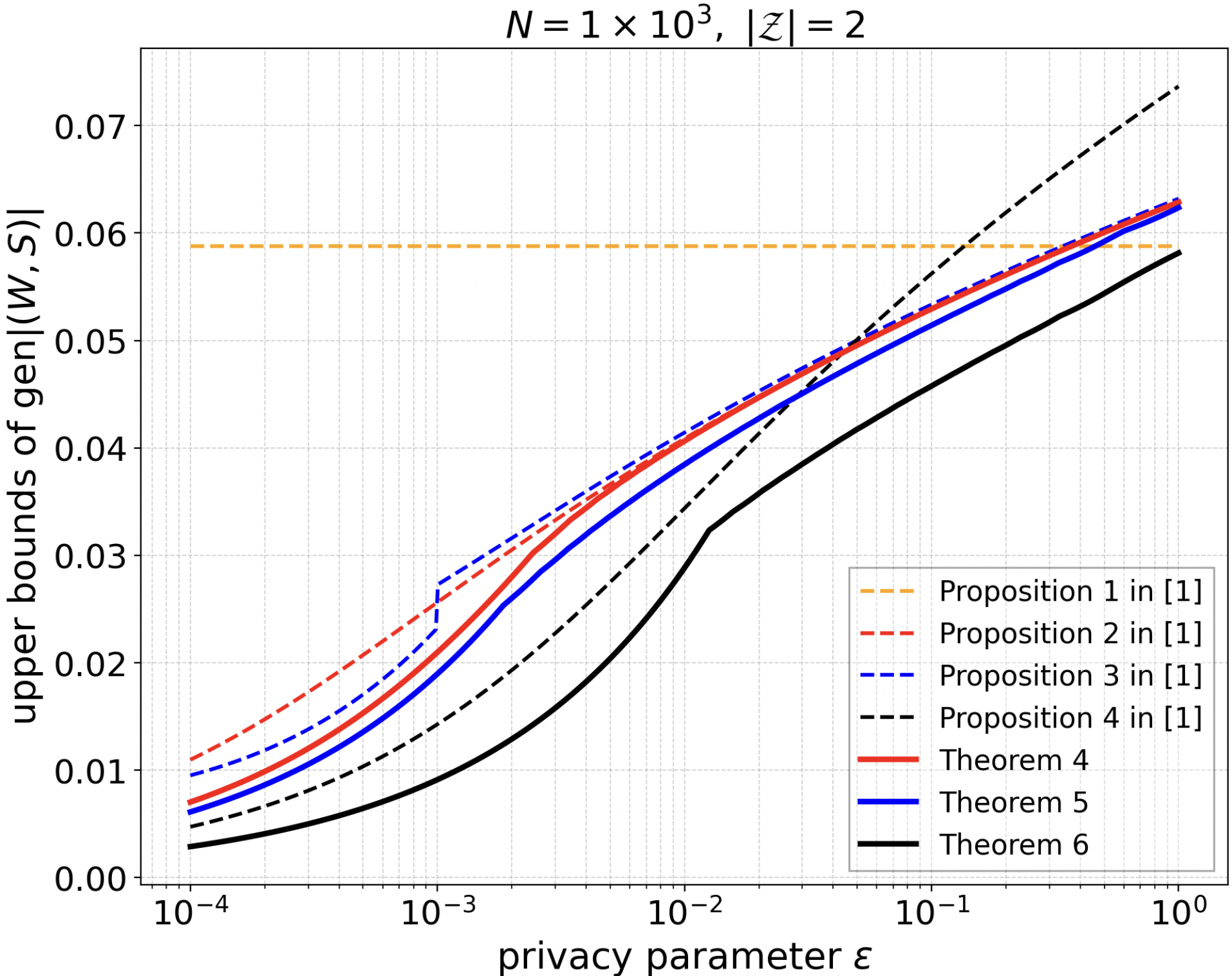}
    \caption{Dataset of size $N =1\times 10^{3}$ with $|\mathcal{Z}| = 2$.} 
    \label{fig::N_3_Z_binary}
\end{figure}
\begin{figure}[htpb]
    \centering
    \includegraphics[scale=0.3]{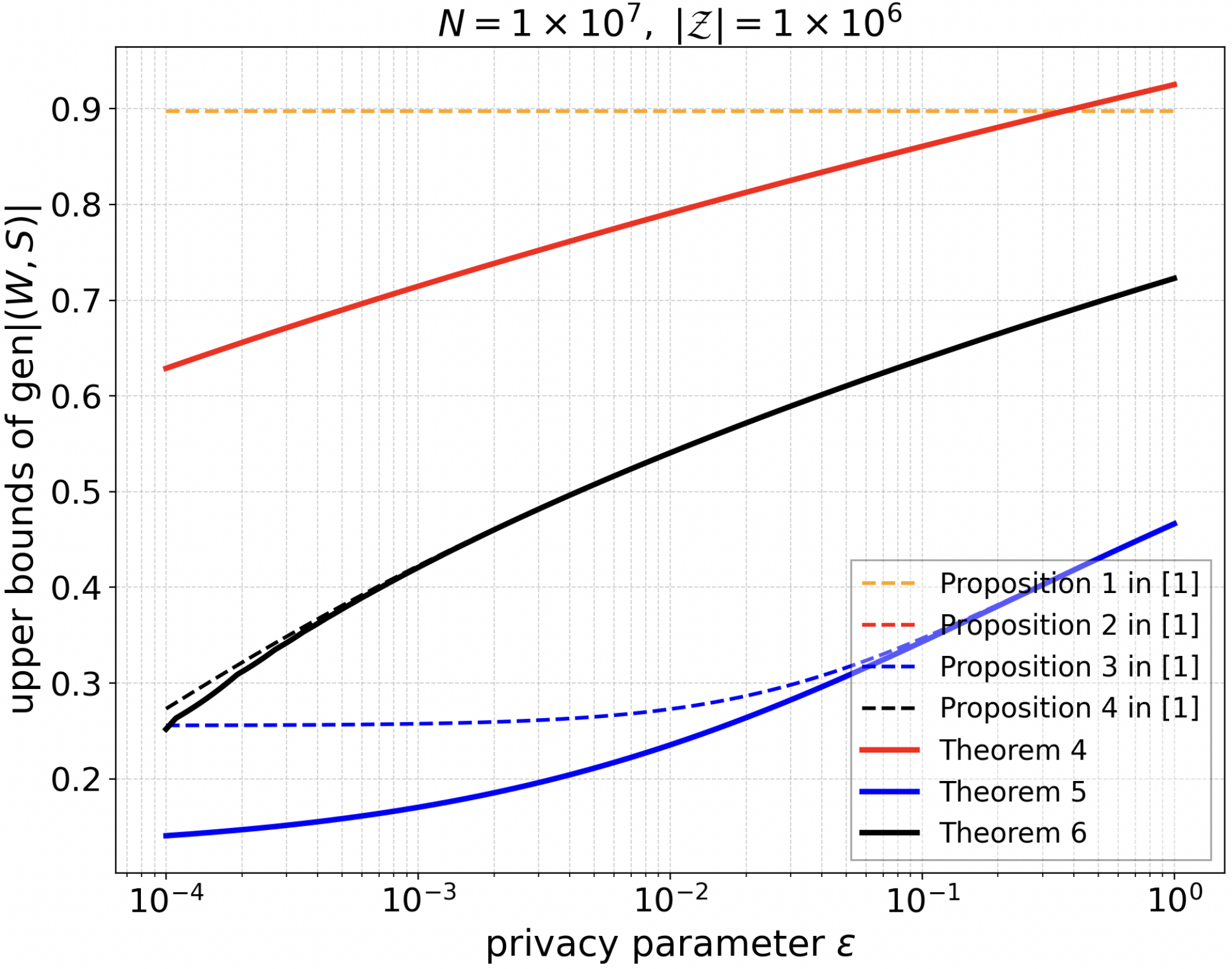}
    \caption{Dataset of size $N = 1\times 10^{7}$ with $|\mathcal{Z}| = 1\times 10^{6}$.} 
    \label{fig::N_7_Z_6}
\end{figure}
\fi

\section{Maximal Leakage Bounds}
\label{sec::ML}

In this section, we present novel generalization bounds that can be efficiently computed by bounding maximal leakage and utilizing~\eqref{eq::gen_bd_prob_ML}, also based on a typicality-based approach. 

\begin{definition}
    Given a joint distribution $P_{SW}$ on finite alphabets $\mathcal{S},\mathcal{W}$, we say the \emph{maximal leakage} from $S$ to $W$ is
    \begin{equation*}
        \mathcal L(S\to W) =\log \sum_{w\in\mathcal W}\max_{s\in\mathcal{S}: P_X(s)>0} P_{W|S}(w|s).
    \end{equation*}
\end{definition}

For discussions on maximal leakage for continuous channels, see~\cite{issa2019operational}. 
Note $\mathcal L(S\to W)$ is the Sibson's mutual information~\cite{verdu2015alpha, sibson1969information}  of order infinity, and it upper bounds mutual information. 
For an $\varepsilon$-DP mechanism, we have~\cite{issa2019operational}
\begin{equation}
    \mathcal{L}(S\rightarrow W)\leq \min \{N\varepsilon, \log|\mathcal{W}|\}. \label{eq::leakage_bd}
\end{equation}

The key advantage of maximal leakage is that it depends on the sample distribution only through its support. Consequently, the corresponding generalization error bounds are independent of the precise distribution over the samples and are particularly amenable to analysis, especially in additive noise settings. 
For generalization bounds via maximal leakage, see~\cite{esposito2021generalization, issa2023generalization}

Moreover, note that the usual way to keep the maximal leakage of an algorithm bounded is to add noise (e.g., Laplacian or Gaussian~\cite{esposito2021generalization}), which aligns with common differential privacy mechanisms~\cite{dwork2006calibrating}. 
Hence, for the study of the generalization error of differentially private algorithms, compared to bounding mutual information, it is more natural to bound the maximal leakage. 
\textcolor{black}{We then introduce our main results. }

\subsection{General Algorithms}

We first provide a general upper bound on maximal leakage, where $\mathcal{A} = P_{W|S}$ is not necessarily differentially private. 

\begin{theorem}
    \label{thm::gen_bd_ML_type}
    Let $S = (Z_1,\ldots Z_N)\in\mathcal{Z}^N$ be a dataset of $N$ instances, each sampled i.i.d. from $P_Z$. 
    Let also $W\in\mathcal{W}$ be a hypothesis obtained with an algorithm $\mathcal{A}$, characterized by $P_{W|S}$. 
    Then the maximal leakage from the dataset to the output hypothesis satisfies 
    \begin{equation*}
        \mathcal{L}(S\rightarrow W)\leq 
        \left(
    |\mathcal{Z}|-1
    \right)\log(N+1).
    \end{equation*}
\end{theorem}

\ifshortver
\else 
The proof of Theorem~\ref{thm::gen_bd_ML_type} can be found in Appendix~\ref{app::gen_bd_ML_type}. 
\fi

Theorem~\ref{thm::gen_bd_ML_type} exactly recovers~\cite[Proposition 1]{rodriguez2021upper} (also see~\eqref{eq::rodri_MI_prop1}) since it is known that maximal leakage upper-bounds mutual information. 

If we utilize the typicality bound in our Claim~\ref{clm::typicality_bd} instead, Theorem~\ref{thm::gen_bd_ML_type} can be strengthened to 
\begin{align*}
    & \mathcal{L}(S\rightarrow W)\\
    & \leq 
    \log \left( \frac{1}{\sqrt{2\pi (|\mathcal{Z}|-1)}}\left(\frac{e}{(|\mathcal{Z}|-1)}\left(N+\frac{|\mathcal{Z}|}{2}\right)\right)^{(|\mathcal{Z}|-1)} \right) \\
    & =  (|\mathcal{Z}|-1)\Bigg(1 + \log\bigg(\frac{N+\frac{|\mathcal{Z}|}{2}}{|\mathcal{Z}|-1}\bigg)\Bigg)
 - \frac{\log \bigl(2\pi (|\mathcal{Z}|-1)\bigr)}{2}.
\end{align*}

Plugging the result into~\eqref{eq::gen_bd_prob_ML}, we derive a more explicit generalization error bound: for every $\eta\in(0,1)$, 
\begin{align*}
    & \mathbf{P}\left(\big|\mathrm{gen}(W,S)\big|\geq\eta \right) \\
    & \leq \frac{2}{\sqrt{2\pi k}}\left(\frac{e}{k}\left(N+\frac{k+1}{2}\right)\right)^k \cdot 
    \exp\left(- \frac{N\eta^2}{2\sigma^2}
    \right).
\end{align*}

\subsection{Differentially Private Algorithms}

We then discuss $\varepsilon$-DP algorithms. 
In the region $1/N<\varepsilon\leq 1$, we derive the following results that improves Theorem~\ref{thm::gen_bd_ML_type}. 

{
\color{black}
\begin{theorem}
    \label{thm::gen_bd_ML_eps_DP}
    Let $S = (Z_1,\ldots,Z_N)$ be a dataset of size $N$ where $Z_i\in\mathcal{Z}$ are sampled i.i.d. from $P_Z$. 
    Let $W\in\mathcal{W}$ be a hypothesis obtained from an $\varepsilon$-DP algorithm $\mathcal{A}$ whose induced distribution is $P_{W|S}$. 
    When $1/N<\varepsilon \leq 1$, 
    \begin{equation*}
                \mathcal{L}(S\rightarrow W) \leq (|\mathcal{Z}|-1) \log\big(
                e(N\varepsilon  + 1)
                \big). 
            \end{equation*}
\end{theorem}
\ifshortver
\else 
The proof of Theorem~\ref{thm::gen_bd_ML_eps_DP} can be found in Appendix~\ref{app::gen_bd_ML_eps_DP}. 
\fi

Theorem~\ref{thm::gen_bd_ML_eps_DP} improves~\eqref{eq::leakage_bd} whenever $|\mathcal Z|-1<\frac{N\varepsilon}{1+\log(N\varepsilon+1)}$.
}

The proof relies on a special choice of representative sets. This approach can be combined with the refined constructions of representative sets developed in the proofs of Theorems~\ref{thm::prop3_plus_overlap} and~\ref{thm::prop4_improved_overlap}. 
Consequently, it is possible to integrate these constructions for mutual information bounds into the current proof of Theorem~\ref{thm::gen_bd_ML_eps_DP} and further strengthen our maximal leakage bounds. 
We leave this direction for future work.

\section{Concluding Remarks}
\label{sec::conclu}

In this paper, we have provided two classes of explicit, easily computable generalization error bounds: one based on mutual information and the other based on maximal leakage.
For the former class, we strictly improved the bounds in~\cite{rodriguez2021upper}, and the comparison is shown in Figure~\ref{fig::N_3_Z_binary} and Figure~\ref{fig::N_7_Z_6}.
Application of the techniques in Section~\ref{sec::MI} to the maximal leakage bounds in Section~\ref{sec::ML} is non-trivial and left for future work.
Moreover, it is still possible to improve bounds in both classes through better choices of the mixture distribution in Lemma~\ref{lem::rodriguez2021upper_KL_bd}. 

In~\cite{rodriguez2021upper}, another class of private algorithms based on Gaussian differential privacy~\cite{dong2022gaussian} has been studied.
It is also of interest to study the generalization error of Gaussian DP algorithms as well, together with other extensions of differential privacy~\cite{mironov2017renyi, kairouz2015composition}.
If specific privatization mechanisms are specified, sharper bounds are expected to be obtainable.

{\color{black}
\section{Acknowledgement}

The authors would like to thank the anonymous reviewers of this article for their valuable suggestions. 

}

\ifshortver
\else
\appendices

\section{Proof of Claim~\ref{clm::typicality_bd}}
\label{app::typicality_bd}

\begin{proof}
    Denote $k := |\mathcal{Z}|-1$, we calculate 
    \begin{align*}
        |\mathcal T_{\mathcal Z,N}|
    & = \binom{N+k}{k} \\
    & = \frac{(N+k)!}{N!\,k!} \\
    & = \frac{\prod_{i=1}^k (N+i)}{k!}
    \\ 
    & \stackrel{(a)}{\leq} \frac{\left(\frac{1}{k}\sum_{i=1}^k (N+i)\right)^k}{k!} \\
    & = \frac{\left(N+\frac{k+1}{2}\right)^k}{k!}\\
    & \stackrel{(b)}{\leq} \frac{\left(N+\frac{k+1}{2}\right)^k}{\sqrt{2\pi k}\,\left(\frac{k}{e}\right)^k} \\
    & = \frac{1}{\sqrt{2\pi k}}\left(\frac{e}{k}\left(N+\frac{k+1}{2}\right)\right)^k , 
    \end{align*}
    where $(a)$ is by the AM-GM inequality applied to $N+1,\ldots,N+k$ and get $\prod^k_{i=1} (N+i)\leq (N + (k+1)/2)^k$, and $(b)$ is by the Stirling's lower bound $k!\geq \sqrt{2\pi k}(k/e)^k$. 
\end{proof}

\section{Proof of Lemma~\ref{lem::dp_overlap_mixture}}
\label{app::proof_dp_overlap_mixture}

\begin{proof}
    Since for every measurable set $E$ we have 
    \begin{equation*}
        Q_b(E)\geq \alpha_{b,i} Q_i(E),
    \end{equation*}
    it follows that
    \begin{align*}
        Q(E) 
        & = \sum_{b=1}^M \omega_b Q_b(E) \\
        & \geq \sum_{b=1}^M \omega_b \alpha_{b,i} Q_i(E) \\
        & =: c_i\, Q_i(E),
    \end{align*}
    where $c_i := \sum_{b=1}^M \omega_b \alpha_{b,i}$.
    Therefore, we have $\log \frac{\mathrm{d}Q_i}{\mathrm{d}Q}\leq -\log c_i$, $Q_i$-a.e.. 
    Since $P\ll Q_i$, we obtain
    \begin{align*}
        D_{\mathrm{KL}}(P\Vert Q) 
        & = \mathbf{E}_P\left[\log \frac{\mathrm{d}P}{\mathrm{d}Q}\right] \\
        & = \mathbf{E}_P\left[\log \frac{\mathrm{d}P}{\mathrm{d}Q_i} +  \log \frac{\mathrm{d}Q_i}{\mathrm{d}Q}\right] \\
        & \leq  \mathbf{E}_P\left[\log \frac{\mathrm{d}P}{\mathrm{d}Q_i} \right] - \log c_i \\
        & =  D_{\mathrm{KL}}(P\Vert Q_i) - \log c_i,
    \end{align*}
    which proves the claim.
\end{proof}

\section{Proof of Theorem~\ref{thm::improved_KL_bd}}
\label{app::improved_KL_bd}

\begin{proof}
Fix $t\in\{1,2,\dots,N\}$. 
Denote the cardinality of alphabet $\mathcal Z$ by $m := |\mathcal Z|$ where $m\geq 2$. 

Construct the representative set $\mathcal{S}_0 =\{s_1,\dots,s_M\}$ with parameter $t$ such that $1\leq t\leq N$. 
Define the uniform mixture
\[
Q_W^{(t)} := \frac{1}{M}\sum_{b=1}^M P_{W|S=s_b}.
\]

Fix an arbitrary dataset $s$.
By the same covering argument used for~\cite[Eq. (52)]{rodriguez2021upper}, there exists some representative $s_i\in\mathcal{S}_0 $ such that
\begin{equation}
\label{eq:cover_used_again}
 \mathsf{d}(s,s_i)\ \le\ \frac{N}{t}(m-1).
\end{equation}

Since $\mathcal A$ is $\varepsilon$-DP, group privacy implies for all measurable $E$,
\[
P_{W|S=s_b}(E)\ \ge\ e^{-\varepsilon  \mathsf{d}(s_b,s_i)}\,P_{W|S=s_i}(E),
\]
we apply our Lemma~\ref{lem::dp_overlap_mixture} with $\omega_b= 1/M$ and $\alpha_{b,i}=e^{-\varepsilon  \mathsf{d}(s_b,s_i)}$, and get 
\begin{align}
& D_{\mathrm{KL}} \bigl(P_{W|S=s}\,\|\,Q_W^{(t)}\bigr)
\nonumber \\
& \le
D_{\mathrm{KL}} \bigl(P_{W|S=s}\,\|\,P_{W|S=s_i}\bigr)
-\log \Big(\frac{1}{M}\sum_{b=1}^M e^{-\varepsilon  \mathsf{d}(s_b,s_i)}\Big). \label{eq:lemma_used_again}
\end{align}

Take $k:= \mathsf{d}(s,s_i)$, by~\eqref{eq::group_DP}, 
\[
D_{\mathrm{KL}} \bigl(P_{W|S=s}\,\|\,P_{W|S=s_i}\bigr)\ \le\ k\varepsilon.
\]
Combining with \eqref{eq:cover_used_again} gives
\begin{equation}
\label{eq:first_term_linear_again}
D_{\mathrm{KL}} \bigl(P_{W|S=s}\,\|\,P_{W|S=s_i}\bigr)
\ \le\ (m-1)\frac{\varepsilon N}{t}.
\end{equation}

For the case of $t=1$, there is only one grid cube, therefore $M=1$ and
$(1/M)\sum_{b=1}^M e^{-\varepsilon  \mathsf{d}(s_b,s_i)}=1$, and hence the last term in  \eqref{eq:lemma_used_again} is $0$.

For the case of $t\geq 2$, we have $M\geq 2$ and there exists some $j\neq i$ with $s_j\in\mathcal{S}_0 $. 
Since for any two size-$N$ datasets $u,v$ one always has $\mathsf{d}(u,v)\le N$.
Therefore,
\begin{align*}
    & \sum_{b=1}^M e^{-\varepsilon  \mathsf{d}(s_b,s_i)} \\
    & \geq e^{-\varepsilon  \mathsf{d}(s_i,s_i)} + e^{-\varepsilon  \mathsf{d}(s_j,s_i)} \\
    & \geq 1 + e^{-\varepsilon N},
\end{align*}
and hence
\begin{equation}
\label{eq:overlap_penalty_again}
-\log \Big(\frac{1}{M}\sum_{b=1}^M e^{-\varepsilon  \mathsf{d}(s_b,s_i)}\Big)
 \le
\log M - \log \bigl(1+e^{-\varepsilon N}\bigr).
\end{equation}

Since  there are at most $t^{|\mathcal{Z}|-1}$ smaller hypercubes, we have $M\le t^{m-1}$, and hence $\log M\le (m-1)\log t$.

Combining \eqref{eq:lemma_used_again}, \eqref{eq:first_term_linear_again}, and \eqref{eq:overlap_penalty_again}, we know for every $s$,
\begin{align*}
   & D_{\mathrm{KL}} \bigl(P_{W|S=s}\,\|\,Q_W^{(t)}\bigr)
\\ 
& \leq  (m-1)\frac{\varepsilon N}{t} + (m-1)\log t - \mathds 1_{\{t\geq 2\}}\log \bigl(1+e^{-\varepsilon N}\bigr).
\end{align*}

One can choose $t$ minimizing the right-hand side and define $Q_W := Q_W^{(t)}$ for that minimizing $t$. 
\end{proof}

\section{Proof of Theorem~\ref{thm::prop3_plus_overlap}}
\label{app::prop3_plus_overlap}

\begin{proof}
Fix $t\in\{1,2,\dots,N\}$.
Let $\mathcal{S}_0 = \{s_1,\dots,s_M\}$ be the representative set constructed as follows: 
instead of constructing a covering of the space $[0,N]^{|\mathcal{Z}| - 1}$ with $t^{|\mathcal{Z}| - 1}$ hypercubes, here we only count the hypercubes strictly needed to cover the hypervolume under the $(|\mathcal{Z}|-2)$-simplex, which has $|\mathcal{Z}| - 1$ perpendicular edges of length $N$ which intersect at the origin. 
By~\cite[Lemma 3]{rodriguez2021upper}, the number of hypercubes needed to cover the hypervolume under the $(|\mathcal{Z}|-2)$-simplex is given by~\eqref{eq::hypercubes needed}.  
From such $\mathcal{S}_0$, we define the uniform mixture
\[
Q_W^{(t)}:=\frac{1}{M}\sum_{b=1}^M P_{W|S=s_b}.
\]

From the same covering argument as in the proof of Theorem~\ref{thm::improved_KL_bd}, there exists an $s_i\in\mathcal{S}_0$ such that
\begin{equation}
\label{eq:cover_radius_prop3_no_m}
\mathsf d(s,s_i)\le (|\mathcal Z|-1)\frac{N}{t}.
\end{equation}

Since $\mathcal A$ is $\varepsilon$-DP, for all measurable $E$,group privacy gives, 
\[
P_{W|S=s_b}(E)\geq e^{-\varepsilon \mathsf d(s_b,s_i)}P_{W|S=s_i}(E).
\]

We apply Lemma~\ref{lem::dp_overlap_mixture} with $\omega_b=1/M$ and $\alpha_{b,i}=e^{-\varepsilon \mathsf d(s_b,s_i)}$, which yields
\begin{align}
& D_{\mathrm{KL}} \bigl(P_{W|S=s}\,\|\,Q_W^{(t)}\bigr)
\nonumber \\
&\le
D_{\mathrm{KL}} \bigl(P_{W|S=s}\,\|\,P_{W|S=s_i}\bigr)
-\log \Big(\frac{1}{M}\sum_{b=1}^M e^{-\varepsilon \mathsf d(s_b,s_i)}\Big).
\label{eq:overlap_lemma_step_no_m}
\end{align}

Take $k:=\mathsf d(s,s_i)$. 
By~\cite[Claim 2]{rodriguez2021upper} we know
\[
D_{\mathrm{KL}} \bigl(P_{W|S=s}\,\|\,P_{W|S=s_i}\bigr)\le k\varepsilon.
\]
Combining with \eqref{eq:cover_radius_prop3_no_m} gives
\begin{equation}
\label{eq:first_term_no_m}
D_{\mathrm{KL}} \bigl(P_{W|S=s}\,\|\,P_{W|S=s_i}\bigr)
\le (|\mathcal Z|-1)\frac{\varepsilon N}{t}.
\end{equation}

If $t=1$, then $M=1$ and the overlap term in \eqref{eq:overlap_lemma_step_no_m} is $0$.

If $t\geq 2$, we have $M\geq 2$ and moreover there exists $j\neq i$ such that
\begin{equation}
\label{eq:neighbor_dist_no_m}
\mathsf d(s_i,s_j)\le \Delta_t:=\Big\lceil\frac{N}{t}\Big\rceil+1.
\end{equation}
Therefore
\begin{align*}
\sum_{b=1}^M e^{-\varepsilon \mathsf d(s_b,s_i)}
& \geq 
e^{-\varepsilon \mathsf d(s_i,s_i)} + e^{-\varepsilon \mathsf d(s_j,s_i)} \\
& \geq 1+e^{-\varepsilon\Delta_t},
\end{align*}
and hence
\begin{equation}
\label{eq:overlap_penalty_no_m}
-\log \Big(\frac{1}{M}\sum_{b=1}^M e^{-\varepsilon \mathsf d(s_b,s_i)}\Big)
\le
\log M - \log \bigl(1+e^{-\varepsilon\Delta_t}\bigr).
\end{equation}

Then by the construction of the representative set described at the beginning of the proof (also see~\cite[Lemma 3]{rodriguez2021upper}), the number of kept hypercubes satisfies $M\le S_{|\mathcal Z|-1}(t)$ where 
\[
S_{|\mathcal Z|-1}(t):=\frac{1}{(|\mathcal Z|-1)!}\frac{(t+|\mathcal Z|-2)!}{(t-1)!}.
\]

We have $\log M\le \log S_{|\mathcal Z|-1}(t)$ and combining \eqref{eq:overlap_lemma_step_no_m}, \eqref{eq:first_term_no_m}, and \eqref{eq:overlap_penalty_no_m} gives
\begin{align*}
 & D_{\mathrm{KL}} \bigl(P_{W|S=s}\,\|\,Q_W^{(t)}\bigr) \nonumber\\
&\le
(|\mathcal Z|-1)\frac{\varepsilon N}{t}
+
\log S_{|\mathcal Z|-1}(t)
-\mathds 1_{\{t\geq 2\}}\log \bigl(1+e^{-\varepsilon\Delta_t}\bigr)
\end{align*}
where
\[
\Delta_t:=\Big\lceil \frac{N}{t}\Big\rceil+1. 
\]

Finally, by~\cite[Lemma 3]{rodriguez2021upper},
\begin{equation}
    S_{|\mathcal Z|-1}(t)\le \frac{1}{(|\mathcal Z|-1)!}\left(t+\frac{|\mathcal Z|-2}{2}\right)^{|\mathcal Z|-1},
\label{eq::hypercubes needed}
\end{equation}
we derive 
\begin{align*}
& D_{\mathrm{KL}} \bigl(P_{W|S=s}\,\|\,Q_W^{(t)}\bigr) \\
&\le
(|\mathcal Z|-1)\frac{\varepsilon N}{t}
+
(|\mathcal Z|-1)\log \left(t+\frac{|\mathcal Z|-2}{2}\right)  \\
&\qquad
-\log (|\mathcal Z|-1)!
-\mathds 1_{\{t\geq 2\}}\log \bigl(1+e^{-\varepsilon\Delta_t}\bigr).
\end{align*}
\end{proof}

\section{Proof of Theorem~\ref{thm::prop4_improved_overlap}}
\label{app::prop4_improved_overlap}

\begin{proof}
The main idea of the proof is to utilize analyses based on strong typicality, as used in~\cite[Appendix E]{rodriguez2021upper}, together with Lemma~\ref{lem::dp_overlap_mixture}.

Recall the definition of \emph{strong typical} set: 
\begin{align*}
    \mathcal T_\varepsilon^N(Z) := \{ Z^N : &\,\,
    \big| T_{Z^N}(a) - P_Z(a)\big| \leq \varepsilon\text{ if } P_Z(a)>0, \\
    & \mathsf{N}(a|Z^N)=0\text{ otherwise}
    \},
\end{align*}
and when the context is clear we use $\mathcal T$ for short. 

By the  union bound and Hoeffding's inequality, we know 
\begin{equation}
    \label{eq:typ_prob}
    P_Z^{\otimes N}(\mathcal{Z}^N \backslash \mathcal{T}) \leq 2|\mathcal{Z}|\exp(-2N\varepsilon^2). 
\end{equation}

Take $m := |\mathcal Z| \ge 2$ for the ease of notation. 

Take $\Delta_t := \lceil \ell_t\rceil+1$ where 
\[
\ell_t := \frac{2\sqrt{N\log N}}{t}.
\]

Take $\eta := \sqrt{(\log N) / N}$. 

By~\cite[Eq. (77)]{rodriguez2021upper}, 
\begin{align}
I(S;W)
&\le
\mathbf{E}\!\left[ D_{\mathrm{KL}}(P_{W|S}\|Q_W^{(t)}) \,\middle|\, S\in\mathcal T\right]\mathbb P(S\in\mathcal T) \nonumber \\
&\,\,\,\,\, +
\mathbf{E}\!\left[ D_{\mathrm{KL}}(P_{W|S}\|Q_W^{(t)}) \,\middle|\, S\notin\mathcal T\right]\mathbb P(S\notin\mathcal T).
\label{eq:decomp}
\end{align}

By~\cite[Eq. (72)]{rodriguez2021upper}, the typical count vector $\mathsf{N}_s$ lies inside an $m$-dimensional hypercube of side at most $2\sqrt{N\log N}$. 
We split each coordinate interval into $t$ parts, giving small hypercubes of side
$\ell_t\le 2\sqrt{N\log N}/t$. 
Take $\mathcal{S}_0=\{s_1,\dots,s_M\}\subset \mathcal T$ to be the set of representative datasets, and define the uniform mixture
\[
Q_W^{(t)} := \frac{1}{M}\sum_{b=1}^M P_{W|S=s_b},
\]
and we know the distance between any $s$ belonging to the $i$-th  hypercube and its center $s_i$ is bounded by 
\begin{equation}
\label{eq:d_to_center}
\mathsf d(s,s_i)\le \frac{m\sqrt{N\log N}}{t}. 
\end{equation}

Fix $s\in\mathcal T$ and let $s_i$ be the representative of its small hypercube, by group privacy
\[
D_{\mathrm{KL}}(P_{W|S=s}\|P_{W|S=s_i})\le \varepsilon\,\mathsf d(s,s_i).
\]

By our Lemma~\ref{lem::dp_overlap_mixture} with $\omega_b=1/M$, 
\begin{align}
& D_{\mathrm{KL}}(P_{W|S=s}\|Q_W^{(t)}) \nonumber \\
&\le
\varepsilon\,\mathsf d(s,s_i)
-\log\!\Big(\frac1M\sum_{b=1}^M e^{-\varepsilon\,\mathsf d(s_b,s_i)}\Big).
\label{eq:overlap_apply}
\end{align}

If $t=1$, then $M=1$ and the logarithmic term in \eqref{eq:overlap_apply} equals $0$.

If $t\ge 2$, since the small hypercubes form a regular grid and we keep only hypercubes that contain at least one type with total count $N$, the kept set satisfies $M\ge 2$ for $t\ge 2$ in the typical region described above.
Hence, for each fixed $i$ we can pick some $j\neq i$ and trivially bound $\mathsf d(s_i,s_j)\le \Delta_t$ with
\[
\Delta_t := \lceil \ell_t\rceil+1
\quad\text{and}\quad
\ell_t=\frac{2\sqrt{N\log N}}{t},
\]
which yields
\[
-\log\!\Big(\frac1M\sum_{b=1}^M e^{-\varepsilon\,\mathsf d(s_b,s_i)}\Big)
\le
\log M-\log(1+e^{-\varepsilon\Delta_t}).
\]

It remains to upper bound $M$. 

At first, since the $m$-dimensional hypercube is partitioned into $t^{m}$ smaller hypercubes, we have $M\le t^{m}$ and hence
\begin{equation}
\label{eq:M_count1}
\log M\le m \log t.
\end{equation}

For each choice of the first $(m-1)$ coordinate bins, the constraint
$\sum_{a\in\mathcal Z} \mathsf{N}(a|s)=N$ restricts the last coordinate to a range of length at most $(m-1)\ell_t$,
so it can intersect at most $m$ bins of length $\ell_t$. Hence the number of kept hypercubes satisfies $M \le m\,t^{m-1}$ and hence 
\begin{equation}
\label{eq:M_count}
\log M\le \log(m\,t^{m-1}).
\end{equation}

Combining \eqref{eq:M_count1} and \eqref{eq:M_count}, we have for $t\ge 2$,
\[
\log M \le \min\{\,m\log t,\ \log(m\,t^{m-1})\,\}.
\]

Therefore, combining \eqref{eq:d_to_center} and \eqref{eq:overlap_apply} gives, for all $s\in\mathcal T$,
\begin{align*}
    & D_{\mathrm{KL}}(P_{W|S=s}\|Q_W^{(t)}) \\
    & \le  
    \frac{\varepsilon m\sqrt{N\log N}}{t}
    -\mathds 1_{\{t\ge 2\}}\log(1+e^{-\varepsilon\Delta_t}), \\
    & \qquad + \mathds 1_{\{t\ge 2\}} 
\min\{ |\mathcal Z|\log t, \log(|\mathcal Z|t^{|\mathcal Z| -  1})\}. 
\end{align*}

For the first term in \eqref{eq:decomp}, taking conditional expectation over $S\in\mathcal T$ gives the same bound. 

For the second term, \cite[Lemma 2]{rodriguez2021upper} and Jensen's inequality  give $D_{\mathrm{KL}}(P_{W|S=s}\|Q_W^{(t)})\le \varepsilon N$ for $s\notin\mathcal T$. 
Multiplied by $\mathbb P(S\notin\mathcal T)\le 2mN^{-2}$ to obtain
\[
\mathbf{E}\!\left[ D_{\mathrm{KL}}(P_{W|S}\|Q_W^{(t)}) \,\middle|\, S\notin\mathcal T\right]\mathbb P(S\notin\mathcal T)
\le \frac{2m\varepsilon}{N}.
\]

Substituting the typical and atypical bounds into \eqref{eq:decomp} yields \eqref{eq:prop4_improved_main}.
\end{proof}



\section{Proof of Theorem~\ref{thm::gen_bd_ML_type}}
\label{app::gen_bd_ML_type}

{
\color{black}

\begin{proof}
Let $T_S$ denote the type of $S$. Since $\mathcal A$ is permutation invariant, two datasets with the same type induce the same output distribution. Hence there exists a conditional distribution $P_{W|T}$ such that
\[
P_{W|S=s}(\cdot)=P_{W|T}(\cdot|T_s), \qquad s\in\mathcal{S},
\]
that is, $W$ depends on $S$ only through $T_S$.

Therefore,
\begin{align*}
    \mathcal{L}(S\rightarrow W) 
    &= \log \sum_{w\in\mathcal W}\max_{s:\,P_S(s)>0} P_{W|S}(w|s) \\
    &= \log \sum_{w\in\mathcal W}\max_{s:\,P_S(s)>0} P_{W|T}(w|T_s) \\
    &= \log \sum_{w\in\mathcal W}\max_{t:\,P_{T_S}(t)>0} P_{W|T}(w|t) \\
    &= \mathcal{L}(T_S\rightarrow W)\\
    & \leq \log |\mathrm{supp}(T_S)|\\
    & \leq \log |\mathcal{T}_{\mathcal{Z}, N}|,
\end{align*}
where the first inequality follows by \cite[Lemma 1]{issa2019operational} and the second inequality follows by $T_S\in \mathcal{T}_{\mathcal{Z}, N}$.

Finally, since $\mathrm{supp}(T_S)\subseteq \mathcal{T}_{\mathcal{Z},N}$ and
\[
|\mathcal{T}_{\mathcal{Z},N}|\leq (N+1)^{|\mathcal Z|-1}
\]
by~\cite[Claim~1]{rodriguez2021upper}, we obtain
\[
\mathcal{L}(S\rightarrow W)
\leq
\log |\mathcal{T}_{\mathcal{Z},N}|
\leq
(|\mathcal Z|-1)\log(N+1).
\]
\end{proof}

}

\section{Proof of Theorem~\ref{thm::gen_bd_ML_eps_DP}}
\label{app::gen_bd_ML_eps_DP}

\begin{proof}
    Fix $t\in \{1,\ldots,N\}$. 
    Similar to~\cite{rodriguez2021upper} we partition the $(|\mathcal{Z}| - 1)$ dimensional simplex by hypercubes of side $1/t$. 
    For each non-empty cell (which contains at least one type $T_s$), we choose one dataset $s_0$ whose type lies in that cell and collect all such representatives in the set $\mathcal{S}_0$. 
    It is straightforward to check that 
    \begin{equation*}
        |\mathcal{S}_0|\leq t^{|\mathcal{Z}|-1}.
    \end{equation*}
    Note if $T_s$ and $T_{s_0}$ are in the same cell, then for coordinates $a = 1,\ldots,|\mathcal{Z}|-1$ we know 
    \begin{equation*}
        |T_s(a) - T_{s_0}(a)|\leq \frac{1}{t}
    \end{equation*}
    and also the last coordinate satisfies 
    \begin{align*}
        |T_s(|\mathcal{Z}|) - T_{s_0}(|\mathcal{Z}|)| & = \left|
        \sum^{|\mathcal{Z}|-1}_{a=1}\big(
        T_s(a) - T_{s_0}(a)
        \big)
        \right| \\
        & \leq  \sum^{|\mathcal{Z}|-1}_{a=1}
        \left|
        T_s(a) - T_{s_0}(a)
        \right| \\
        & \leq \frac{|\mathcal{Z}|-1}{t}.
    \end{align*}
    Therefore 
    \begin{align*}
        \Vert T_s - T_{s_0} \Vert_1 &\leq  \sum^{|\mathcal{Z}|}_{a=1}
        \left|
        T_s(a) - T_{s_0}(a)
        \right| \\
        & \leq \frac{2(|\mathcal{Z}|-1)}{t}.
    \end{align*}
    
    Note we use the multiset distance~\cite{rodriguez2021upper}: 
    \begin{align*}
        \mathsf{d}(s, s_0) & = \frac{1}{2}\sum_{a\in\mathcal{Z}}\left|\mathsf{N}(a|s) - \mathsf{N}(a|s_0)\right| \\
        & = \frac{N}{2}\Vert T_s - T_{s_0} \Vert_1\\
        & \leq \frac{N}{2} \cdot \frac{2(|\mathcal{Z}|-1)}{t} \\
        & = \frac{N}{t}(|\mathcal{Z}|-1).
    \end{align*}

    For each $s\in\mathcal{S}$, let $s_0(s)\in\mathcal{S}_0 $ be the representative whose type lies in the same cell as $T_s$,
    and define $k(s):=d(s,s_0(s))$. By the covering argument above, for all $s$,
    \[
    k(s)\le K := \frac{N}{t}(|\mathcal Z|-1).
    \]
    
    Hence, for all $s\in\mathcal{S}$ and all $w\in\mathcal W$,
    \begin{align*}
    P_{W|S=s}(w)
    &\le e^{\varepsilon k(s)}P_{W|S=s_0(s)}(w)\\
    &\le e^{\varepsilon k(s)}\max_{\bar s\in\mathcal{S}_0 }P_{W|S=\bar s}(w)\\
    &\le e^{\varepsilon K}\max_{\bar s\in\mathcal{S}_0 }P_{W|S=\bar s}(w).
    \end{align*}
    Taking the maximum over $s$ yields
    \[
    \max_{s\in\mathcal{S}}P_{W|S=s}(w)
    \le e^{\varepsilon K}\max_{\bar s\in\mathcal{S}_0 }P_{W|S=\bar s}(w).
    \]

    {\color{black}
    Summing over $w$, we can write \begin{align*}
\sum_{w\in\mathcal W} \max_{s\in\mathcal S} P_{W|S=s}(w)
&\le e^{\varepsilon K}\sum_{w\in\mathcal W}\max_{\bar s\in\mathcal S_0} P_{W|S=\bar s}(w) \\
&\le e^{\varepsilon K}\sum_{w\in\mathcal W}\sum_{\bar s\in\mathcal S_0} P_{W|S=\bar s}(w) \\
&= e^{\varepsilon K}\sum_{\bar s\in\mathcal S_0}\sum_{w\in\mathcal W} P_{W|S=\bar s}(w) \\
&= e^{\varepsilon K}\,|\mathcal S_0|.
\end{align*}

}
    Therefore, by the definition of maximal leakage,
    \begin{align*}
    \mathcal L(S\to W)
    &=\log \sum_{w\in\mathcal W}\max_{s\in\mathcal{S}} P_{W|S=s}(w)\\
    &\le \varepsilon K + \log|\mathcal{S}_0 |\\
    &\le \varepsilon \frac{N}{t}(|\mathcal Z|-1) + (|\mathcal Z|-1)\log t \\
    &=(|\mathcal Z|-1)\left(\frac{N\varepsilon}{t}+\log t\right).
    \end{align*}

    We then optimize over $t$, and obtain the following result:
    \begin{itemize}
        \item When $\varepsilon\leq 1/N$, 
            \begin{equation*}
                \mathcal{L}(S\rightarrow W) \leq (|\mathcal{Z}|-1) N\varepsilon,
            \end{equation*}
        which is actually no stronger than~\eqref{eq::leakage_bd}. 

        \item When $1/N<\varepsilon \leq1$, 
            \begin{equation*}
                \mathcal{L}(S\rightarrow W) \leq (|\mathcal{Z}|-1) \log\big(
                e(N\varepsilon  + 1)
                \big),
            \end{equation*}
        which is tighter when $|\mathcal Z|-1<\frac{N\varepsilon}{1+\log(N\varepsilon+1)}$.
    \end{itemize}
\end{proof}
\fi

\ifshortver
\else
\bibliographystyle{IEEEtran}
\bibliography{ref}
\fi

\end{document}